\documentclass[11pt]{article}
\usepackage[utf8]{inputenc}
\usepackage{hyperref}
\usepackage{ifthen}
\usepackage{enumerate,fullpage}
\usepackage{amssymb}
\usepackage{amsmath}
\usepackage{amsthm}
\usepackage{graphicx}
\usepackage{version}
\usepackage{multirow}
\usepackage{float}
\usepackage{epstopdf}
\usepackage{caption}
\usepackage{dcolumn}
\usepackage{algorithm}
\usepackage{algorithmic}
\usepackage{setspace}
\usepackage{makecell}
\usepackage{caption}
\usepackage{natbib}
\usepackage{dsfont}

\newcounter{statement}

\newtheorem{definition}[statement]{Definition}

\newtheorem{proposition}[statement]{Proposition}
\newtheorem{remark}[statement]{Remark}

\providecommand{\keywords}[1]
{
  \small        
  \textbf{\textbf{Keywords:}} #1
}

\hypersetup{
        bookmarks=true,         % show bookmarks bar?
        unicode=false,          % non-Latin characters in Acrobat’s bookmarks
        pdftoolbar=true,        % show Acrobat’s toolbar?
        pdfmenubar=true,        % show Acrobat’s menu?
        pdffitwindow=false,     % window fit to page when opened
        pdfstartview={FitH},    % fits the width of the page to the window
        pdftitle={Enabling Trade-offs in Machine Learning-based Matching for Refugee Resettlement},    % title
        pdfauthor={Nils Olberg},     % author
        pdfsubject={},   % subject of the document
        pdfcreator={},   % creator of the document
        pdfproducer={}, % producer of the document
        pdfkeywords={}, % list of keywords
        pdfnewwindow=true,      % links in new PDF window
        colorlinks=true,       % false: boxed links; true: colored links
        linkcolor=black,          % color of internal links (change box color with linkbordercolor)
        citecolor=black,        % color of links to bibliography
        filecolor=magenta,      % color of file links
        urlcolor=cyan           % color of external links
}

\begin{document}

\title{Enabling Trade-offs in Machine Learning-based Matching\\for Refugee Resettlement\footnote{The research presented in this paper was performed in early 2019. We thank the participants of the INFORMS Workshop on Market Design 2019 for their feedback. Note that the paper by \cite{acharya_combining_2022} studies the same problem and obtains similar results. The papers were written independently from each other, and we only found out about the existence of the paper by \cite{acharya_combining_2022} in July 2019 after we started circulating the first version of the present paper online. However, given the large overlap between the two papers and given that the work by \cite{acharya_combining_2022} was already more mature than ours, we decided to retire this research project in the fall of 2019.}}

\author{Nils Olberg\\University of Zurich\\olberg@ifi.uzh.ch 
\and Sven Seuken\\University of Zurich\\seuken@ifi.uzh.ch}

\date{}

\maketitle

\spacing{1.5}

\begin{abstract}
The Swiss State Secretariat for Migration recently announced a pilot project for a machine learning-based assignment process for refugee resettlement. This approach has the potential to substantially increase the overall employment rate of refugees in Switzerland. However, the currently proposed method ignores families' preferences. In this paper, we build on this prior work and propose two matching mechanisms that additionally take families' preferences over locations into account. The first mechanism is strategyproof while the second is not but achieves higher family welfare. Importantly, we parameterize both mechanisms, giving placement officers precise control how to trade off family welfare against overall employment success. Preliminary simulations on synthetic data show that both mechanisms can significantly increase family welfare even with only a small loss on the overall employment rate of refugees.
\end{abstract}


\section{Introduction}

Refugee families seeking shelter in Switzerland are currently assigned randomly to one of the 26 Swiss cantons according to a distribution key unless their asylum claim is rejected within three months after arrival. This practice ignores potential synergies between refugees and cantons, suggesting suboptimal integration outcomes. \cite{bansak_improving_2018} proposed a machine learning-based algorithm for family placement that aims to optimize the overall employment rate of refugees. The algorithm, which could increase the employment rate of refugees in Switzerland by about 73\%, works in three stages. In the first stage, a machine learning model predicts, for each refugee-location pair, the probability that a refugee will find employment at the corresponding location. In the second stage, these individual refugee probabilities are transformed to a family-level metric. In the last stage, solving an optimization problem provides the final assignment.
The Swiss State Secretariat for Migration (SEM), which is responsible for assigning refugees to cantons, has recently announced to test this machine learning-based assignment process in a pilot project.

Although finding employment is important for successful integration, there are many reasons why refugee resettlement procedures should also enable families to express their individual preferences over resettlement locations, even if these preferences are potentially in conflict with employment success. \cite{delacretaz_matching_2020} argue that refugee families themselves know best where they are likely to thrive. \cite{jones_local_2018} point out that resettlement systems which ignore families' preferences are disempowering for refugees, and suggest that giving them a say in the assignment process could increase their well-being. As described by \cite{jones_matching_2017}, ignoring families' individual preferences has even caused families seeking shelter in Finland to cancel their asylum applications. This is tragic, considering the reasons why families abandon their home countries.

To address this issue, we propose two mechanisms that build on the machine learning-based approach by \cite{bansak_improving_2018}, but now also take family welfare into account (in terms of the refugees' reported preferences). We show that our first mechanism, the \emph{constrained random serial dictatorship mechanism (CRSD)} is family-strategyproof, i.e., families cannot benefit from misreporting their true preferences. The second mechanism, the \emph{constrained rank value mechanism (CRV)}, is not strategyproof, but is always  weakly better in terms of family welfare. While prior work has already considered incorporating refugee preferences into the resettlement process (e.g., \cite{delacretaz_matching_2020}) the primary innovation of our paper is that the mechanisms we design are parameterized, giving placement officers precise control over the impact of families' preferences on the final matching. To the best of our knowledge, prior mechanisms for refugee resettlement were either optimization-based or preference-based. Our mechanisms combine both concepts.
In our simulations (using synthetic "proof-of-concept" data), we show that, for both mechanisms, the trade-off parameter can be chosen in such a way that family welfare is significantly improved (compared to the one-sided assignment mechanism) with only a minimal loss on overall employment success. 

\section{Preliminaries}\label{sec:model}

We consider a set of refugee families $F$ and a set of locations $L$. A distribution key (which typically depends on the population of each location) determines a quota $q_j \in \mathbb{N}$ for each location $j \in L$, which is the number of families that location $j$ is obligated to host. We assume that $q = (q_j)_{j \in L}$ is chosen in such a way that each family can be assigned to exactly one location, i.e., it holds that $\sum_{j \in L} q_j = |F|$. Location $j$ derives its preference for family $i$ from a predicted integration success $\pi_{ij} \in [0,1]$. We assume that the predicted values for integration success are provided to us by machine learning models trained on historic resettlement data. As is done in \cite{bansak_improving_2018} and \cite{ahani_placement_2021}, we assume that  $\pi_{ij}$ corresponds to the predicted probability that family $i$ will find employment at location $j$.
For now, we assume that family $i \in F$ has a weak preference order $\succeq_i$ over locations in $L$. In the appendix we present an extension for the case of incomplete preferences. Let $\succeq_{ij}$ denote the position of location $j\in L$ in preference order $\succeq_i$. 
A matching is a mapping $\mu: F \rightarrow L \cup \{\emptyset\}$, and $F(\mu) = \{ i \in F \mid \mu(i) \neq \emptyset \} \subseteq F$ is the set of assigned families under matching $\mu$. Further, let $F_j(\mu) = \{ i \in F \mid \mu(i) = j \}$ denote the set of all families matched to location $j$ under $\mu$. A matching is feasible if $|F_j(\mu)| = q_j$ for all $j \in L$.

We call the expected number of successfully integrated families, i.e., $z(\mu) = \sum\limits_{i \in F} \pi_{i\mu(i)}$, the \textit{government objective} of $\mu$. A feasible matching $\mu^*$ is \textit{government-optimal} if $\mu^* \in \arg\max_{\mu} \{ z(\mu) \}$. Our goal is to find a feasible matching $\mu$ that (1) maximizes family welfare in terms of reported preferences $\succeq$ $=$ $(\succeq_i)_{i \in F}$ and at the same time (2) ensures that $z(\mu)$ is within a factor of $\alpha$ of $z(\mu^*)$ for a previously chosen $\alpha \in [0,1]$.

\begin{definition}
Let $\alpha \in [0,1]$. A feasible matching $\mu$ is an $\alpha$-approximation of the government-optimal matching if $z(\mu) \ge \alpha z^*$, where $z^* = \max_{\mu} \{ z(\mu) \}$.
\end{definition}

We measure family welfare using the average rank $\rho(\mu) = \frac{1}{|F|} \sum_{i \in F} \succeq_{i \mu(i)}$ of matching $\mu$.
Additionally, we use the \emph{cumulative rank distribution} $\Delta(\mu)$, which provides more detailed information about the goodness of matching $\mu$ in terms of family welfare than $\rho(\mu)$.

\begin{definition}
Let $\delta_k(\mu)$ denote the number of families that are assigned to their $k$-th choice under matching $\mu$. The cumulative rank distribution of $\mu$ is a vector $\Delta(\mu) \in \mathbb{N}^{|L|}$, where the $k$-th entry $\Delta_k(\mu)$ denotes the number of families that are matched to their $k$-th or better choice under $\mu$, i.e., $\Delta_k(\mu) = \sum_{k' = 1}^k \delta_k(\mu)$. It holds that $\Delta_{|L|}(\mu) = |F_\mu^*|$.
\end{definition}

\begin{remark}
Requiring refugees to report a preference order over locations (or even assuming that they have preformed preferences) may be problematic in some countries. For example, there can be hundreds of potential resettlement locations for a family seeking shelter in the United Kingdom (\cite{jones_local_2018}). It is hard to imagine that refugees can come up with a complete preference order over that many options. In Switzerland, however, refugees can only be matched to one out of 26 cantons. Further, note that we are not requiring families to have complete preference orders over locations (see appendix). In practice, families could be provided with information on cantons after arrival, which could help them to form at least incomplete preference orders.
\end{remark}

\begin{remark}
We agree with \cite{jones_local_2018} that efficiency and strategyproofness are more important than stability in the context of refugee resettlement mechanisms. Therefore, we will not analyze our mechanisms in terms of stability in this paper.\end{remark}

\section{Mechanisms}

In the following subsections we describe two mechanisms that consider families' preferences. The first mechanism, the \textit{constrained random serial dictatorship mechanism} (CRSD), is family-strategyproof while the second, the \textit{constrained rank value mechanism} (CRV), is not. However, CRV will usually achieve higher family welfare than CRSD. 
Both mechanisms allow placement officers to choose a parameter $\alpha \in [0,1]$, which ensures that the computed matching is an $\alpha$-approximation of the government-optimal matching.  

\subsection{Constrained Random Serial Dictatorship Mechanism}

\begin{algorithm}[t]
    \caption{Constrained Random Serial Dictatorship (CRSD)}
    \label{alg:crsd}
    \begin{algorithmic}[1]
    \REQUIRE $F$, $L$, $\succeq$, $\pi$, $b$, $\alpha$
    \STATE $\mu(i) := \emptyset$ for all $i \in F$.
    \STATE Compute the objective value $z^*$ of an optimal solution to IP$_\mu$(\ref{ip:crsd}).\label{line-CRSD-ip}
    \STATE $Q := F$
    \WHILE{$Q \neq \emptyset$}
                \STATE Remove randomly chosen $i \in Q$ from $Q$.
                \WHILE{$\succeq_i \neq \emptyset$}\label{line:while-incomplete}
                        \STATE Let $j$ denote $i$'s current top choice in $\succeq_i$ and remove $j$ from $\succeq_i$.
                        \IF{$|F_j(\mu)| < q_j$}\label{line-CRSD-check-feasible}
                                \STATE $\mu' := \mu $; $\mu'(i) := j$
                                \STATE Solve IP$_{\mu'}$(\ref{ip:crsd}) and let $z'$ denote the objective value of the solution.\label{line-CRSD-ip2}
                                \IF{$z' \ge \alpha z^*$}
                                        \STATE $\mu(i) := j$
                                        \STATE \textbf{break}
                                \ENDIF
                        \ENDIF 
                \ENDWHILE
    \ENDWHILE
    \ENSURE $\mu$
\end{algorithmic}
\end{algorithm}

The constrained random serial dictatorship mechanism (CRSD) is a constrained version of the well-known random serial dictatorship mechanism. The general idea of CRSD is to let a family only choose their match from the set of remaining locations if it can be guaranteed that an $\alpha$-approximation of the government-optimal matching is still achievable. Algorithm \ref{alg:crsd} provides a detailed description of the mechanism.
 
Initially, the algorithm computes the objective value $z(\mu^*)$ of a government-optimal solution $\mu^*$. Afterwards, it initializes an (infeasible) empty matching $\mu$. Families are then sorted in a random order and processed sequentially. When it is family $i$'s turn to choose among the remaining locations, $i$ is only assigned to $j$ under $\mu$ if it can be guaranteed that an $\alpha$-approximation of $\mu^*$ is still achievable.

The integer program IP(\ref{ip:crsd}) has to be solved in Line \ref{line-CRSD-ip} and Line \ref{line-CRSD-ip2} of the mechanism. IP(\ref{ip:crsd}) ensures that the intermediate matching $\mu$ is preserved. If the objective value of an optimal solution to IP(\ref{ip:crsd}) is smaller than $\alpha z(\mu^*)$, then $i$ is not allowed to be matched to $j$.

\begin{flalign}
\text{IP}_\mu \text{(\ref{ip:crsd}):}& & \text{maximize } & \displaystyle \sum_{i \in F} \sum_{j \in L} \pi_{ij} x_{ij} &\\
&  &\text{subject to } & \displaystyle \sum\limits_{i \in F} x_{ij} = q_j & \forall j \in L \label{ip:crsd-capacity} &\\
                                 & & & \displaystyle \sum\limits_{j \in L} x_{ij} = 1 & \forall i \in F \label{ip:crsd-assigned} &\\
                 & & & x_{i\mu(i)} = 1 & \forall i \in F(\mu) \label{ip:crsd-preserve} &\\
                 & & & x_{ij} \in \{0,1\} & \forall i \in F \text{, } \forall j\in L &
\end{flalign}\label{ip:crsd}

An optimal solution of IP(\ref{ip:crsd}) induces a matching that maximizes the overall predicted employment rate. Variable $x_{ij}$ indicates whether family $i$ will be assigned to location $j$.
Constraints (\ref{ip:crsd-capacity}) ensure that every location hosts as many families as required for a feasible matching. Constraints (\ref{ip:crsd-assigned}) guarantee that every family is assigned to exactly one location. Constraints (\ref{ip:crsd-preserve}) preserve the intermediate matching $\mu$.

\begin{proposition}
CRSD is family-strategyproof, and the matching computed by CRSD is an $\alpha$-approximation of the government-optimal matching.
\end{proposition}

\begin{proof}
Because of the feasibility check in Line \ref{line-CRSD-ip2}, we know that at each step of the algorithm there exists a feasible $\alpha$-approximation $\mu'$ that preserves the intermediate matching $\mu$. Thus, the final matching is an $\alpha$-approximation of the government-optimal matching.
Further, CRSD is family-strategyproof because family $i$ cannot influence which locations will remain available to it once it is $i$'s turn to choose, and by stating its true preferences it is guaranteed that the best among the remaining locations is chosen.
\end{proof}

\begin{remark}[Computational Complexity]
In its original formulation, IP(\ref{ip:crsd}) boils down to a maximum-weight matching problem. Finding a solution to this problem can be done in polynomial time, e.g., using the Hungarian method (\cite{kuhn_hungarian_1955}). As we will see in Section \ref{add_constraints}, additional constraints (e.g., service constraints, capacity constraints, etc.) could easily be integrated in the CRSD mechanism. However, this transforms the maximum-weight matching problem into a NP-hard problem, which can significantly increase the overall runtime of the algorithm.
\end{remark}

\subsection{Constrained Rank Value Mechanism}

Before introducing the constrained rank value mechanism (CRV), we need to establish the concept of a rank value function. 
Along the lines of \cite{featherstone_rank_2020}, we use rank value functions to assign values between 0 and 1 to positions in preference orders.

\begin{definition}
A rank value function is a mapping $v: \{1,...,|L|\} \rightarrow [0,1]$ that is monotonically decreasing.
\end{definition}

\begin{definition}
Given a rank value function $v$, a set of families $F$, a set of locations $L$, quotas $q$, preference orders $\succeq$, predicted employment probabilities $\pi$, and a lower bound $\gamma$, the \textit{constrained maximum rank value problem (CMRV)} is to find a feasible matching $\mu$ that maximizes $\sum_{i \in F} \sum_{j \in L} v(\succeq_{ij}) x_{ij}$, such that $z(\mu) \ge \gamma$.
\end{definition}

By solving an instance of CMRV with $\gamma = \alpha z(\mu^*)$, we can find an $\alpha$-approximation of the government-optimal matching that maximizes family welfare in terms of $v$.

\begin{proposition}
The constrained maximum rank value problem is NP-hard.
\end{proposition}

\begin{proof}
Suppose that we are given an instance $I=(N,w,a,b)$ of the Knapsack problem, where $N=\{1,...,n\}$ is the set of items, $w_i \in \mathbb{R}_{\ge 0}$ is the value for item $i$, $a_i \in \mathbb{R}_{\ge 0}$ is the size of item $i$, and $b \in \mathbb{R}_{\ge 0}$ is the capacity of the knapsack. Construct a CMRV instance $\hat I=(F,L,q,\pi,\succeq,\gamma)$ as follows. Let $F=\{f_1,...,f_n\} \cup \{\bar f_1,...,\bar f_n\}$ and $L=\{\ell_1,...,\ell_n\} \cup \{\bar\ell_1,...,\bar\ell_n\}$. Without loss of generality assume that $w_1 \ge ... \ge w_n$ and $\sum_{i=1}^n a_i = \frac{1}{4n}$. Further, assume that $b < \frac{1}{4n}$. Otherwise we would have a trivial instance where the optimal solution is to put all items in the knapsack. 
Set the rank value function to be 
$$v(i)=
\begin{cases}
w_i / 2 & 1 \le i \le n\\
0 & n+1 \le i \le 2n
\end{cases}.$$
All families $f_i$ have the same preference order
$$\ell_1 \succeq_{f_i} ... \succeq_{f_i} \ell_n \succeq_{f_i} \bar\ell_1 \succeq_{f_i} ... \succeq_{f_i} \bar\ell_n,$$ 
and all families $\bar f_i$ have the preference order 
$$\bar\ell_1 \succeq_{\bar f_i} ... \succeq_{\bar f_i} \bar\ell_n \succeq_{\bar f_i} \ell_1 \succeq_{\bar f_i} ... \succeq_{\bar f_i} \ell_n.$$ 
Set $\pi_{f_i \bar\ell_i} = a_i + \frac{1}{4n}$ and $\pi_{f_i \ell_i} = \pi_{\bar f_i \ell_i} = \pi_{\bar f_i \bar\ell_i} = \frac{1}{4n}$ for all $i \in N$. For all other family-location pairs $(i,j)$ set $\pi_{ij} = 0$. Finally, let each location have a capacity of $1$ and choose $\gamma = \frac{2n + 1}{4n} - b$. 

Note that, due to the choice of $\pi$, $f_i$ can only be matched to either $\ell_i$ or $\bar\ell_i$ in any feasible solution for $\hat I$. Otherwise the $\gamma$-constraint would be violated. The same holds for $\bar f_i$. Further, because of the capacity constraints, for any feasible matching $\mu$ it holds that $ \mu(f_i) = \ell_i \Leftrightarrow \mu(\bar f_i) = \bar \ell_i.$

Let $\mu^*$ denote an optimal solution for $\hat I$. Using $\mu^*$, we can obtain an optimal solution $x[\mu^*]\in \{0,1\}^n$ for $I$, where $x[\mu^*]_i = 1$ corresponds to item $i$ being placed in the knapsack, by setting $x[\mu^*]_i = 1$ if and only if $\mu^*(f_i) = \ell_i$ for all $i$. We call $x[\mu^*]$ the solution induced by $\mu^*$.
We now show that $x[\mu^*]$ is an optimal solution for $I$.

First, observe that a matching $\mu$ is a feasible solution for $\hat I$ if and only if the induced solution $x[\mu]$ is feasible for $I$:
\begin{align}
& & \sum_{i=1}^n \pi_{i \mu(i)} & \ge \gamma \label{proof_feasible_crvm_start} \\
& \Leftrightarrow & \sum_{i=1}^n \frac{1}{4n} + \frac{1}{4n} + a_i (1-x[\mu]_i) & \ge \frac{2n+1}{4n} - b\\
& \Leftrightarrow & \sum_{i=1}^n - a_i x[\mu]_i & \ge - b\\
& \Leftrightarrow & \sum_{i=1}^n a_i x[\mu]_i & \le b. \label{proof_feasible_crvm_end}
\end{align}

Second, for any feasible matching $\mu$ it holds that
\begin{align}
z(\mu) & = \sum_{i=1}^n \frac{w_i}{2} \mathds{1}[\mu(f_i) = \ell_i] + \frac{w_i}{2} \mathds{1}[\mu(\bar f_i) = \bar\ell_i]\\
& = \sum_{i=1}^n w_i \mathds{1}[\mu(f_i) = \ell_i]\\
& = \sum_{i=1}^n w_i x[\mu]_i,
\end{align}

where the second equality comes from the fact that $\mu$ is feasible and thus $\mu(f_i)= \ell_i \Leftrightarrow \mu(\bar f_i)= \bar \ell_i$.
It follows that that $x[\mu^*]$ is an optimal solution for $I$.

\end{proof}

CMRV can be formulated as an integer program.

\begin{flalign} 
\text{IP} \text{(\ref{ip:rv})}:& & \text{maximize } & \displaystyle \sum_{i \in F} \sum_{j \in L} v(\succeq_{ij}) x_{ij} &\\
&  &\text{subject to } & \displaystyle \sum\limits_{i \in F} x_{ij} = q_j & \forall j \in L \label{ip:rv-capacity} &\\
                                 & & & \displaystyle \sum\limits_{j \in L} x_{ij} = 1 & \forall i \in F \label{ip:rv-assigned} &\\
                 & & & \displaystyle \sum\limits_{i \in F} \sum\limits_{j \in L} \pi_{ij} x_{ij} \ge \gamma & \label{ip:rv-approx} &\\
                 & & & x_{ij} \in \{0,1\} & \forall i \in F \text{, } \forall j\in L &                 
\end{flalign}\label{ip:rv}

The objective function of IP(\ref{ip:rv}) maximizes family welfare in terms of the rank value function $v$.
Analogous to IP(\ref{ip:crsd}), Constraints (\ref{ip:rv-capacity}) ensure that every location hosts as many families as required for a feasible matching, and Constraints (\ref{ip:rv-assigned}) ensure that every family is assigned to exactly one location. Constraint (\ref{ip:rv-approx}) is required to guarantee that $z(\mu) \ge \gamma$.

Suppose that we have a predefined rank value function $v$, e.g., $v(k) = \frac{1}{k}$.
The constrained rank value mechanism, described in Algorithm \ref{alg:rv}, computes a family-optimal matching $\mu$ according to $v$, such that $z(\mu) \ge \alpha z(\mu^*)$.

\begin{algorithm}[h]
    \caption{Constrained Rank Value Mechanism (CRV)}
    \label{alg:rv}
    \begin{algorithmic}[1]
    \REQUIRE $F$, $L$, $\succeq$, $\pi$, $b$, $\alpha$
    \STATE Compute the objective value $z(\mu^*)$ of an optimal solution to IP(\ref{ip:crsd}).
    \STATE Let $\mu$ denote the matching induced by the solution of IP(\ref{ip:rv}) with $\gamma = \alpha z(\mu^*)$.\label{line:solve-rv}
    \ENSURE $\mu$
\end{algorithmic}
\end{algorithm}

In contrast to CRSD, it is possible to construct instances where families can benefit from misreporting their true preference orders under CRV. However, a manipulation strategy is not straightforward since a refugee family would need to have at least some knowledge about the predictions $\pi_{ij}$ of the machine learning models or the government-optimal matching $\mu^*$ and the preference orders of other families. It has to be further investigated whether families could in practice exploit this weakness of CRV.

\begin{proposition}
CRV is not family-strategyproof, and CRV always produces an $\alpha$-approximation of the government-optimal matching.
\end{proposition}

\begin{proof}

Constraint (\ref{ip:rv-approx}) guarantees that the matching $\mu$ computed by CRV is an $\alpha$-approximation of the government-optimal matching. To see that CRV is not family-strategyproof, we refer the reader to \cite{featherstone_rank_2020}.
\end{proof}

\subsection{Possible Extensions} \label{add_constraints}

The model introduced in Section \ref{sec:model} is rather simple and does not necessarily capture all constraints imposed on feasible matchings in the real world. As described by \cite{delacretaz_matching_2020}, an agency responsible for refugee resettlement might have to incorporate family sizes, i.e., $q_j$ denotes the number of refugees instead of families a location is obligated to host. These constraints would introduce additional combinatorial complexity to the problem. 
Another potential modification is replacing Constraints (\ref{ip:crsd-capacity}) and Constraints (\ref{ip:rv-capacity}) respectively by capacity constraints, i.e., interpreting $q_j$ as an upper bound.
Similarly, one could introduce additional service constraints, e.g., constraints concerning housing or medical conditions of refugee families, as described by \cite{delacretaz_matching_2020} and \cite{ahani_placement_2021}.
All these restrictions (and others) can easily be incorporated in CRSD and CRV by adding appropriate constraints to the IP formulations IP(\ref{ip:crsd}) and IP(\ref{alg:rv}). 

\section{Simulations}

In order to compare the performance of CRSD and CRV, we run simulations on randomly generated instances. For our simulations, we assume that both mechanisms have access to the true preferences of families. Notice, however, that this assumption might be unreasonable when these mechanisms would be used in real-world applications, especially in the case of CRV since CRV is not family-strategyproof.

All simulations were run on a laptop computer with an Intel(R) Core(TM) i7-8550U CPU 1.80GHz processor and 16GB RAM running Ubuntu 18.04.

\subsection{Instance Generation}\label{generation}

We use the following approach to generate instances. Each instance consists of a total of 100 families and 26 locations. There are four types of refugee families $f_1, f_2, f_3, f_4$ and four types of locations $\ell_1,\ell_2,\ell_3,\ell_4$. The predicted employment probabilities of family types at location types are uniformly distributed according to Table \ref{table:gen-pi}. 

\begin{table}[H]
\begin{center}
\begin{tabular}{ |l|c|c|c|c| }
\hline
\makecell{$\omega^\pi_{f\ell}$} & \makecell{$\ell_1$ (1)} & \makecell{$\ell_2$ (9)} & \makecell{$\ell_3$ (6)} & \makecell{$\ell_4$ (10)}\\
\hline
$f_1 (15)$ & 0.6 & 0.5 & 0.5 & 0.3\\
$f_2 (25)$ & 0.3 & 0.4 & 0.2 & 0.1\\
$f_3 (20)$ & 0.3 & 0.2 & 0.4 & 0.1\\
$f_4 (40)$ & 0.1 & 0.1 & 0.1 & 0.1\\
\hline
\end{tabular}
\end{center}
\caption {Values for generating predicted employment probabilities. For family $i$ of type $f$ and location $j$ of type $\ell$, $\pi_{ij} \sim \mathcal{U}(0, \omega^\pi_{f\ell})$.}\label{table:gen-pi}
\end{table}

The numbers in brackets indicate for each family type (location type) how many families (locations) of that type are present in an instance.
Families of type $f_1$ have high predicted employment probability for each of the four location types except for locations of type $\ell_4$. Type $f_2$ families on the other hand are less likely to be employed in locations of type $\ell_1$, even less likely in locations of type $\ell_3$ and $\ell_4$, and have highest probability of employment for locations of types $\ell_2$. Type $f_3$ is similar to $f_2$, except that for those families the predicted employment probabilities for type $\ell_2$ locations and type $\ell_3$ locations are swapped. Families of type $f_4$ have a low predicted employment probability at all locations.

Locations of type $\ell_1$ ($\ell_2$ and $\ell_3$) are obligated to host 4 (2) times more families than locations of type $\ell_4$.
The preference orders of families are derived from randomly generated valuation functions $(u_i)_{i \in F}$ according to Table \ref{table:gen-val}.

\begin{table}[H]
\begin{center}
\begin{tabular}{ |l|c|c|c|c| }
\hline
\makecell{$\omega^u_{f\ell}$} & \makecell{$\ell_1$ (1)} & \makecell{$\ell_2$ (9)} & \makecell{$\ell_3$ (6)} & \makecell{$\ell_4$ (10)}\\
\hline
$f_1 (15)$ & 1.0 & 0.6 & 0.6 & 0.3\\
$f_2 (25)$ & 0.8 & 1.0 & 0.6 & 0.3\\
$f_3 (20)$ & 0.8 & 0.6 & 1.0 & 0.3\\
$f_4 (40)$ & 1.0 & 0.6 & 0.6 & 0.3\\
\hline
\end{tabular}
\end{center}
\caption {Values for generating preference orders. For family $i$ of type $f$ and location $j$ of type $\ell$, $u_{ij} \sim \mathcal{U}(0, \omega^u_{f\ell})$.}\label{table:gen-val}
\end{table}

\subsection{Mechanism Performances}

We include Top Trading Cycles (TTC) and Deferred Acceptance (DA) as benchmarks in our simulations. Priorities of locations over families are derived by sorting families according to $\pi_{ij}$ in decreasing order.

Figure \ref{fig:average_rank_complete} illustrates the performance of CRSD, CRV, TTC and DA on 20 randomly generated instances.

\begin{figure}[H]
\centering
\includegraphics[scale=0.25]{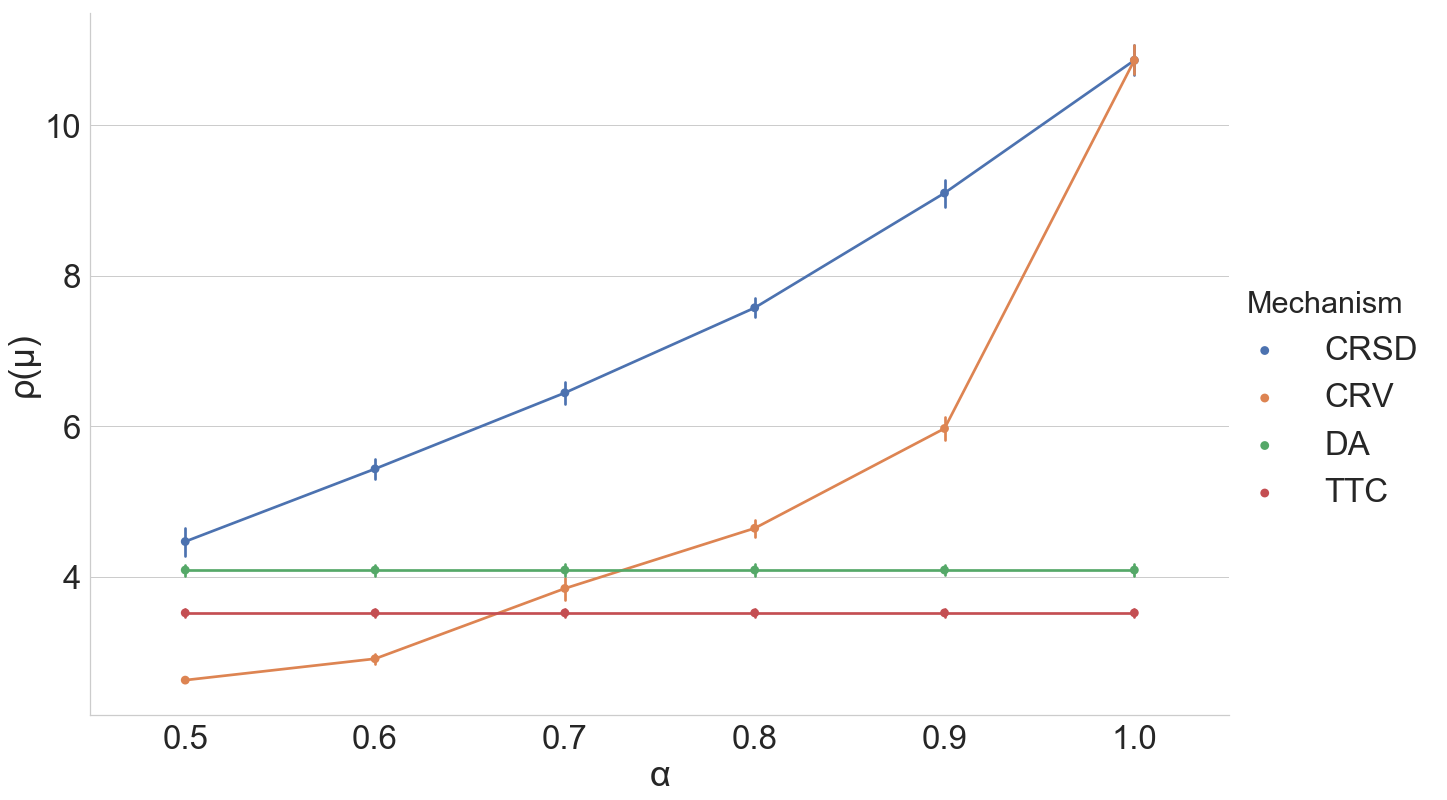}
\caption{Performance of CRSD, CRV, TTC and DA in terms of family welfare.}
\label{fig:average_rank_complete}
\end{figure}

\begin{figure}[H]
\centering
\includegraphics[scale=0.25]{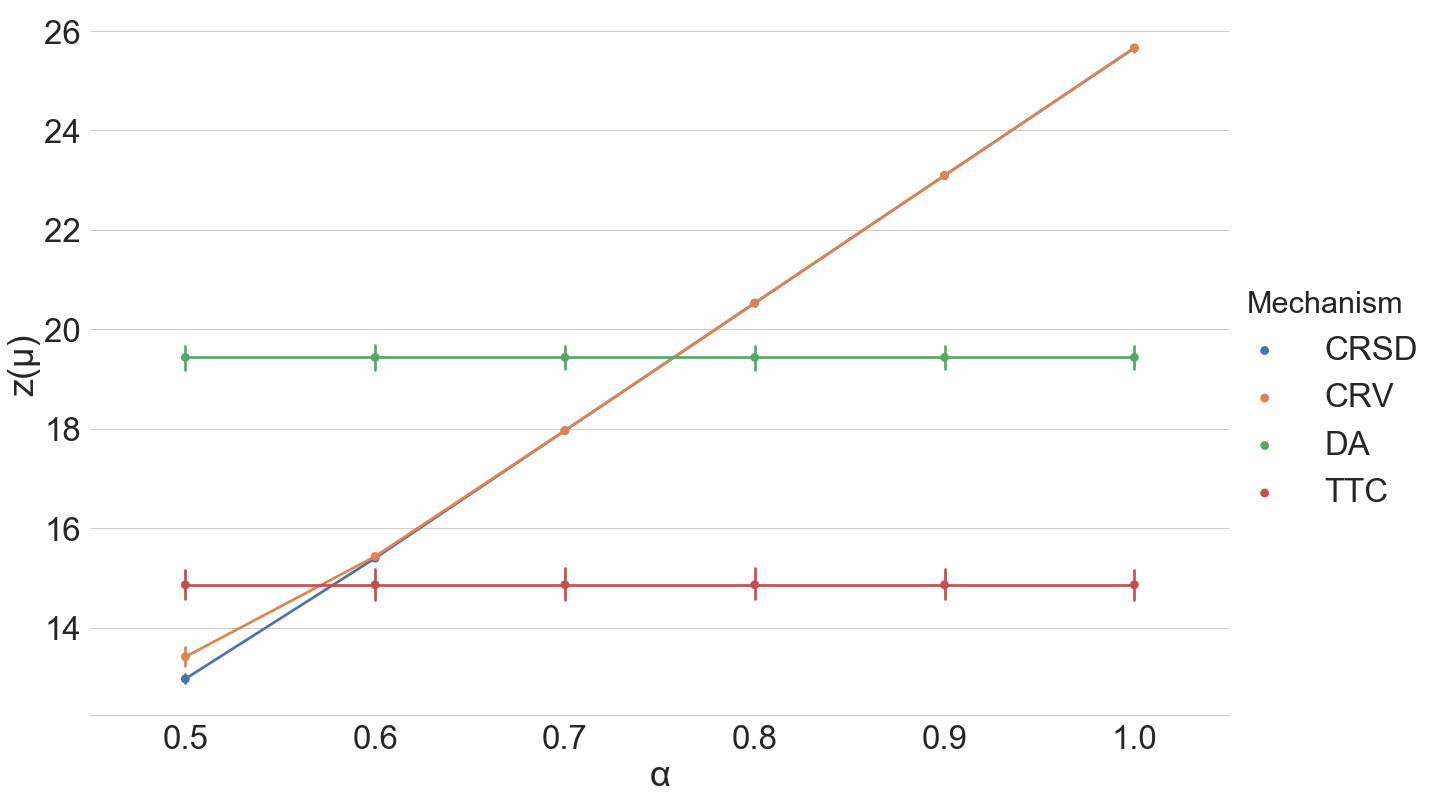}
\caption{Performance of CRSD, CRV, TTC and DA in terms of employment probability.}
\label{fig:employment_complete}
\end{figure}

For CRSD and CRV the average rank $\rho(\mu)$ strictly decreases (and thus family welfare strictly increases) with $\alpha$ going to 0.
Keep in mind that for $\alpha = 1$, both CRSD and CRV produce a government-optimal matching. 
The average rank of CRV matchings strictly dominates the average rank of CRSD matchings for fixed $\alpha$, which is what we would expect. However, CRSD easily outperforms the government-optimal matching, even for values of $\alpha$ close to 1. 
In this concrete setting, $\rho(\mu)$ can be decreased by almost 2 for CRSD and almost 5 for CRV, even if placement officers are only willing to sacrifice 10\% of the predicted overall employment rate. 
Although both TTC and DA achieve high family welfare, remember that these mechanisms cannot give any guarantees in terms of the government objective, as is shown by Figure \ref{fig:employment_complete}.

\begin{figure}[H]
\centering
\includegraphics[scale=0.25]{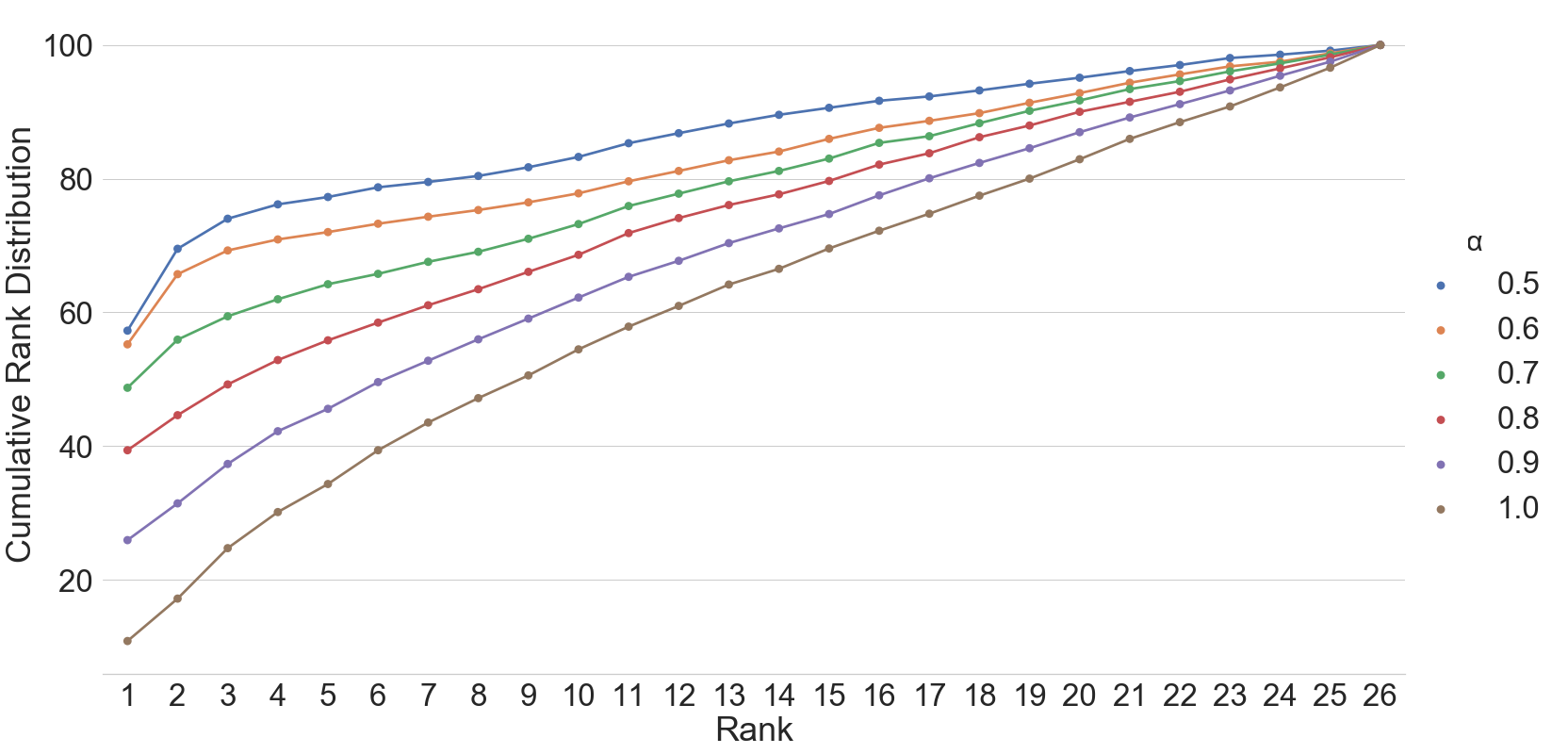}
\caption{Cumulative rank distribution of CRSD.}
\label{fig:cum_rank_crsd}
\end{figure}

\begin{figure}[H]
\centering
\includegraphics[scale=0.25]{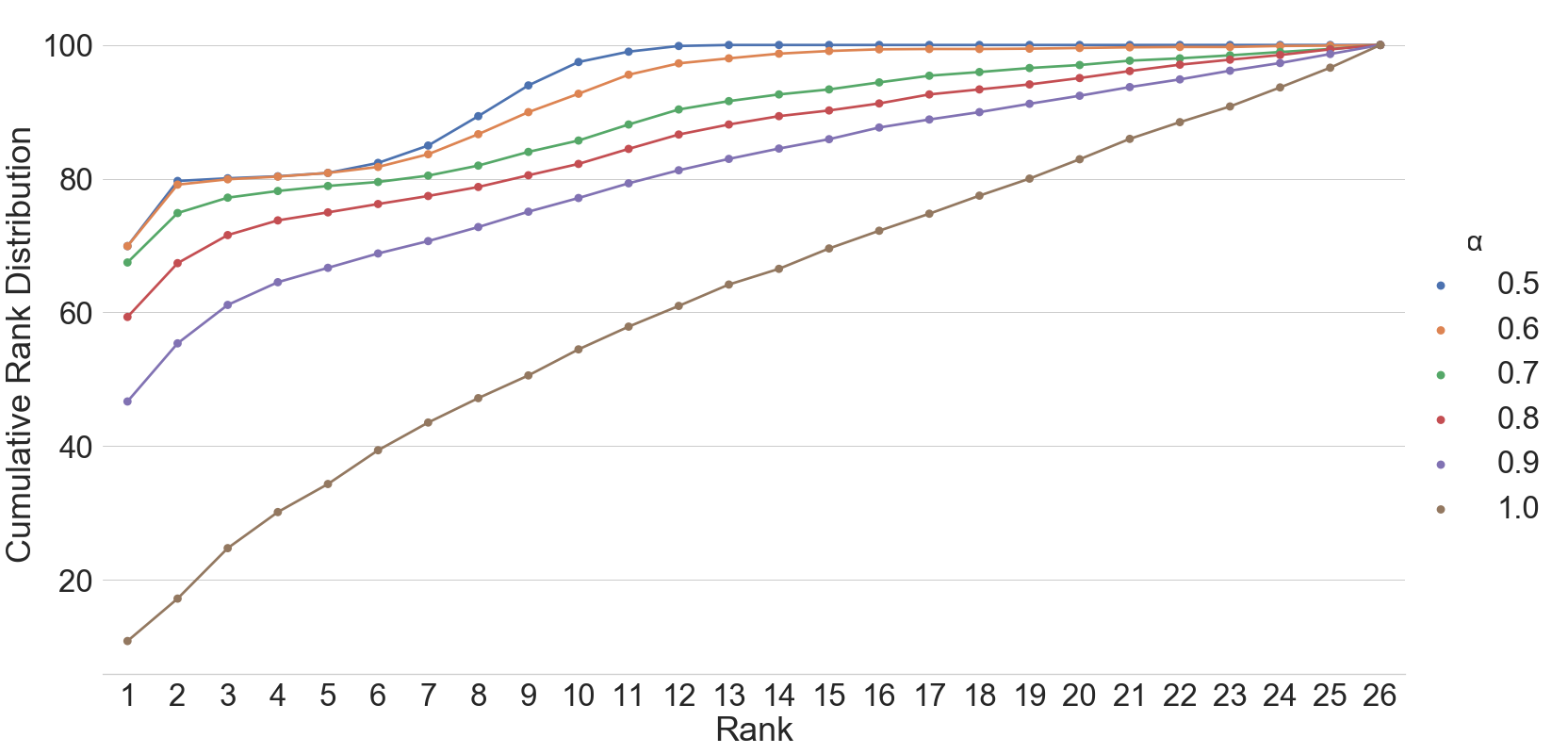}
\caption{Cumulative rank distribution of CRV.}
\label{fig:cum_rank_crv}
\end{figure}

The cumulative rank distributions of CRSD and CRV provide a more detailed description on how family welfare improves with smaller values of $\alpha$. 
As Figure \ref{fig:cum_rank_crsd} and Figure \ref{fig:cum_rank_crv} show, the number of families that are assigned to their first choice can be increased by a factor of roughly 2.4 using CRSD and by a factor of roughly 4.3 using CRV by giving up just 10\% of the government-optimal solution value.

\section{Conclusion}

In this paper, we have proposed two mechanisms, CRSD and CRV, capable of considering families' preferences while simultaneously respecting tight lower bounds on the overall predicted employment rate. While CRSD is family-strategyproof, our simulations show that CRV is in general superior in terms of family welfare when families have complete preference orders over locations. In the case of incomplete preference orders, these results are qualitatively the same and even more pronounced (see Section \ref{incomplete_preferences} in the appendix for details).

Both mechanisms require refugee families to have previously formed preferences over resettlement locations. Helping families form beliefs over where they are most likely to thrive should have a positive impact on resettlement outcomes. Therefore, future work should explore how families can extend incomplete preference orders. For example, a system which asks refugees to rank properties of cantons could derive preference orders for them (see \cite{delacretaz_matching_2020}). Additionally, both mechanisms should be evaluated on real-world data.

\section*{Acknowledgements}

We thank Stefania Ionescu and Jakob Weissteiner for insightful discussions. This paper is part of a project that has received funding from the European Research Council (ERC) under the European Union’s Horizon 2020 research and innovation programme (Grant agreement No. 805542).

\bibliographystyle{apalike}
\bibliography{references}

\section{Appendix}

\subsection{Incomplete Preference Orders}\label{incomplete_preferences}

We also compare the performance of CRSD and CRV in a setting where families only have incomplete preference orders, which is closer to a real-world setting. Here, we use a slightly different way to measure family welfare. Let $F_\mu^*$ denote the set of families that were matched to one of their ranked locations under $\mu$, i.e., $F_\mu^* = \{ i \in F \mid \mu(i) \in L_i \}$, where $L_i \subseteq L$ denotes the set of locations ranked by family $i$. When we compute the average rank $\rho(\mu)$ of a matching, we only consider families in $F_\mu^*$. 
Because $\rho(\mu)$ then only captures family welfare for families in $F_\mu^*$, we also look at $\tau(\mu)= |F \setminus F_\mu^*|$, which is the number of families that were not matched to any location in their preference order.

When families only have incomplete preferences orders, CRSD and CRV have to be slightly modified. When CRSD is run on instances with incomplete preference orders, it can happen that a family remains unmatched after the while loop in Line \ref{line:while-incomplete}. In our simulations, all these unmatched families are simply assigned government-optimally. In the case of CRV, we simply have to exclude a family-location pair $(i,j)$ from the objective function if $j$ was not ranked by $i$.

A families' preference order is generated just as in the complete preferences setting, but is cut off after position $\kappa_i$, where $\kappa_i$ is sampled from a $\Gamma(2,1.5)$-distribution. Because TTC and DA -- at least in their original design -- do not necessarily produce feasible matchings when families only have incomplete preferences, they are excluded from our analysis in this setting.

\begin{figure}[H]
\centering
\includegraphics[scale=0.25]{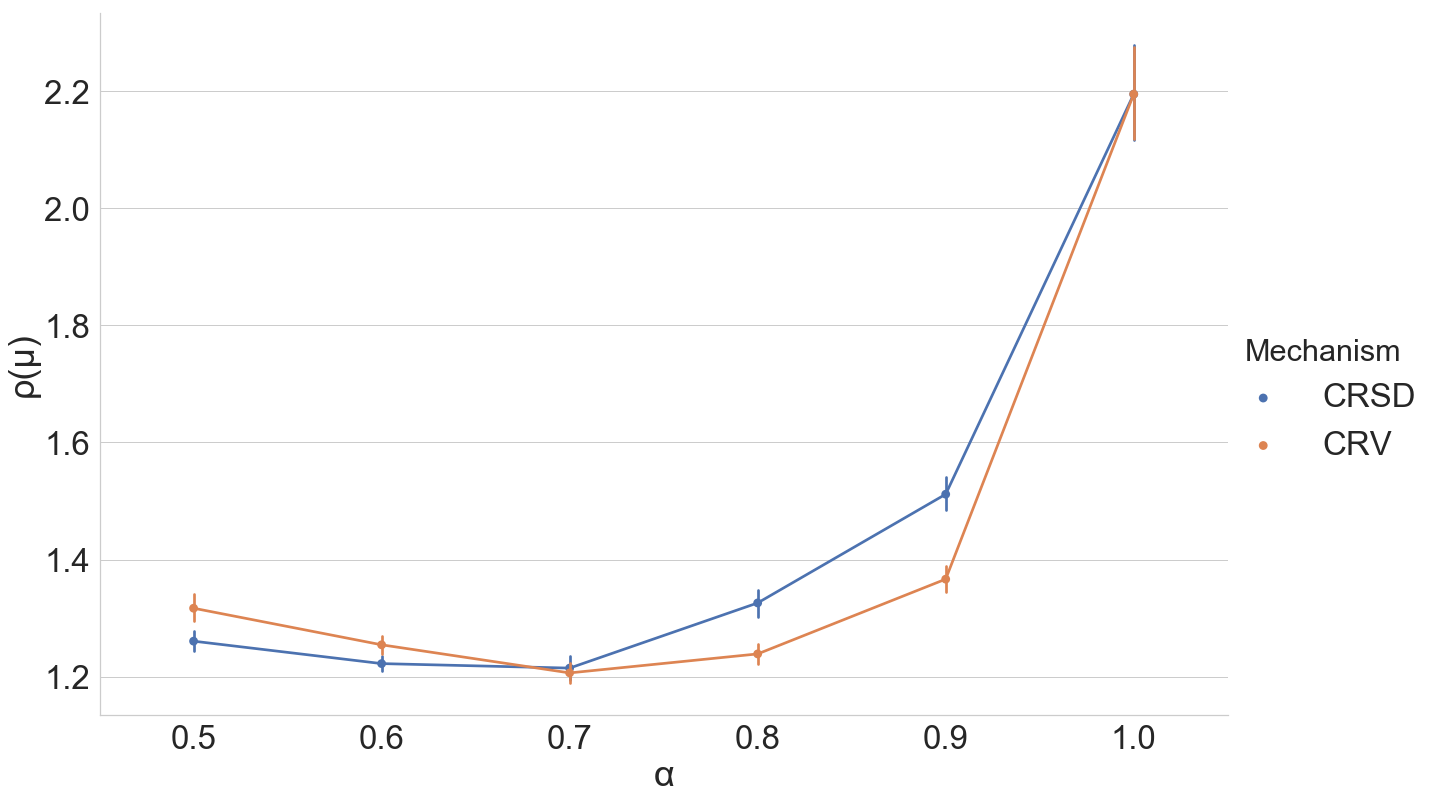}
\caption{$\rho(\mu)$ of CRSD and CRV for different values of $\alpha$.}
\label{fig:average_rank_incomplete}
\end{figure}

\begin{figure}[H]
\centering
\includegraphics[scale=0.25]{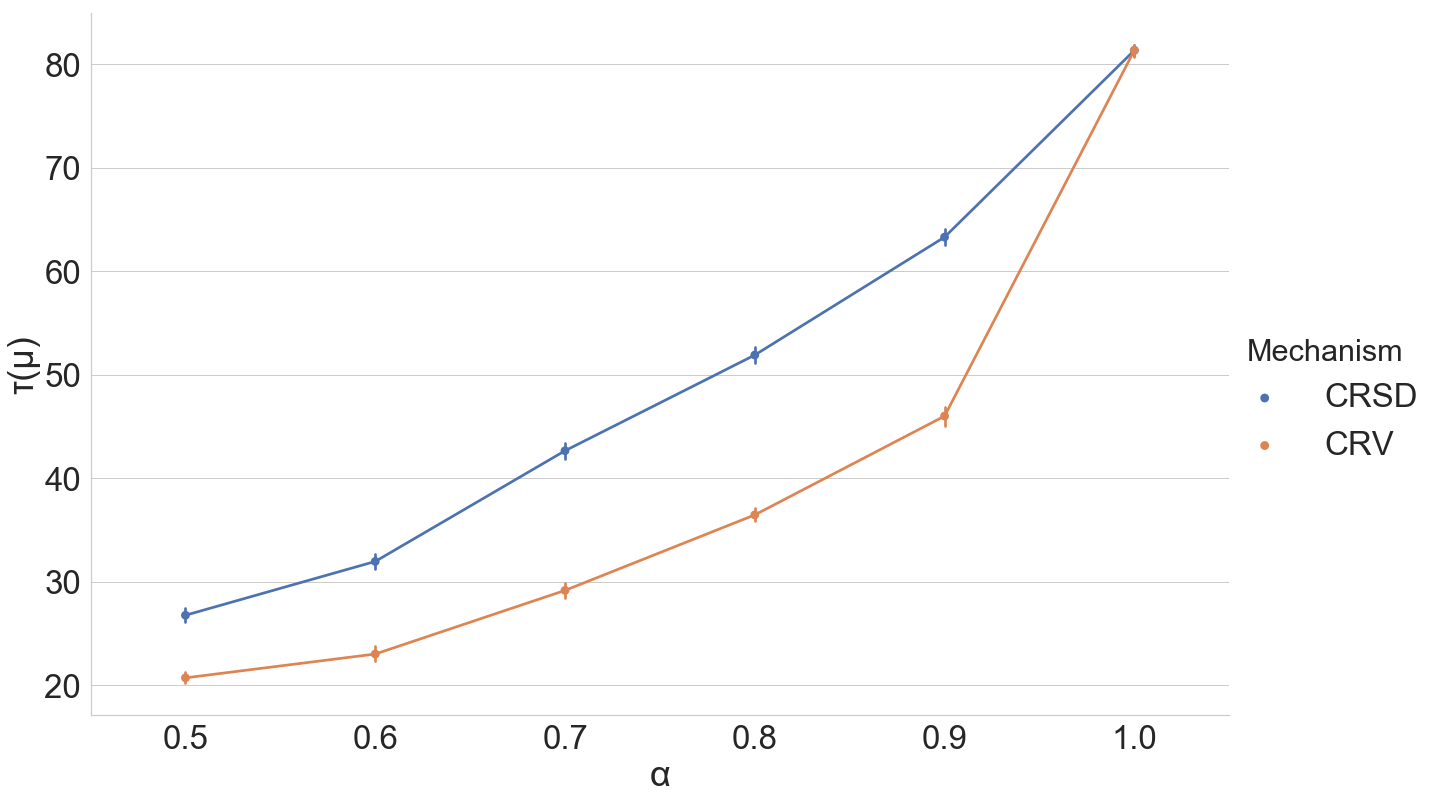}
\caption{$\tau(\mu)$ of CRSD and CRV for different values of $\alpha$.}
\label{fig:employment_incomplete}
\end{figure}

Again, we observe that the average rank strictly decreases for values of $\alpha$ close to 1. Decreasing the value of $\alpha$ for $\alpha \in \{0.5, 0.6, 0.7\}$ slightly increases the average rank again, which can be explained by a further reduction of $\tau(\mu)$, i.e., the number of families that are not assigned to a location in their preference order continues to go down.

\begin{figure}[H]
\centering
\includegraphics[scale=0.25]{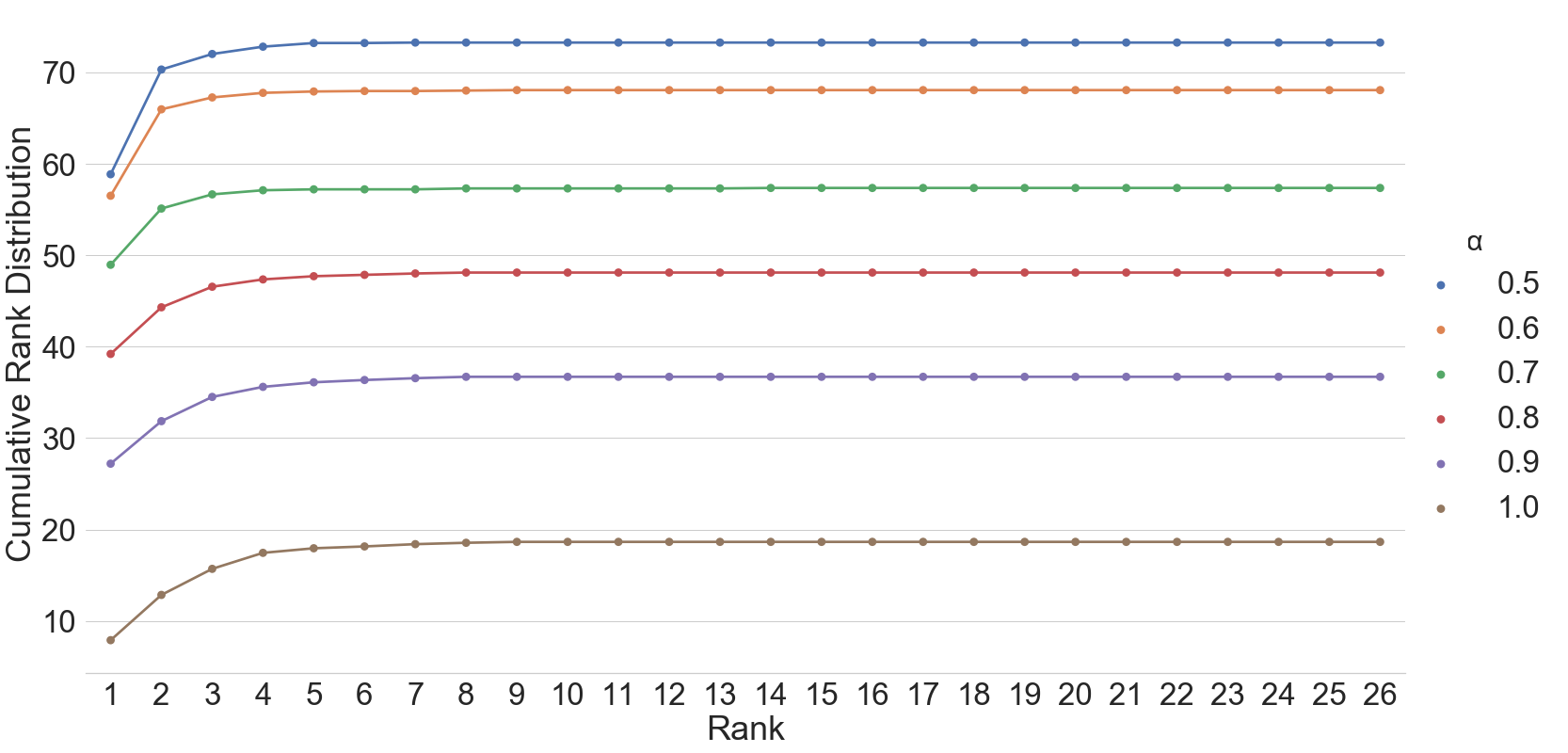}
\caption{Cumulative rank distribution of CRSD.}
\label{fig:cum_rank_crsd_incomplete}
\end{figure}

\begin{figure}[H]
\centering
\includegraphics[scale=0.25]{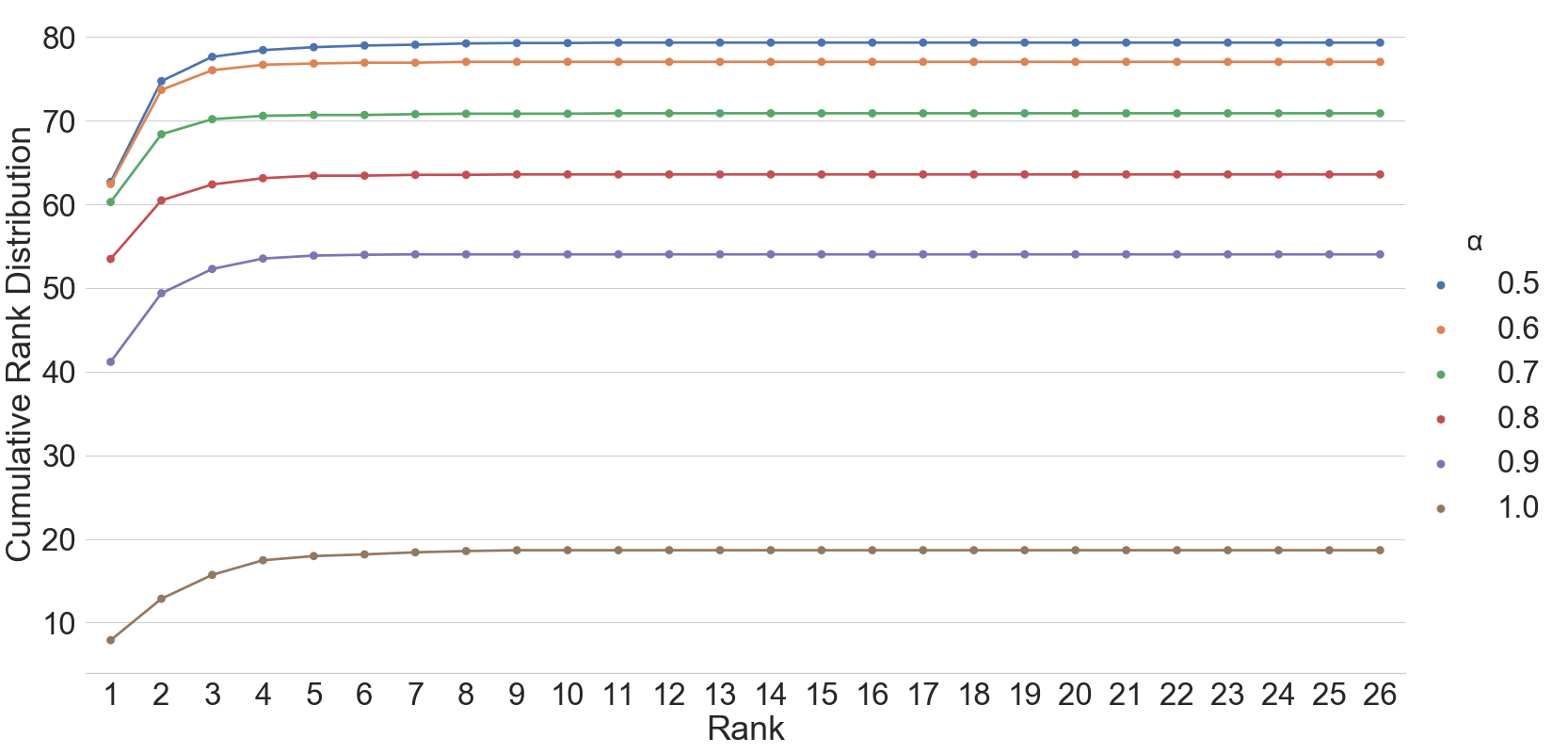}
\caption{Cumulative rank distribution of CRV.}
\label{fig:cum_rank_crv_incomplete}
\end{figure}

Choosing $\alpha = 0.9$ instead of $\alpha = 1$ increases $\Delta_1(\mu)$ from an average of $7.9$ to $27.2$, which corresponds to a factor of roughly $3.4$. Unsurprisingly, the effect is even stronger for CRV, where $\Delta_1(\mu)$ increases to $41.1$ (a factor of roughly $5.2$). Those numbers continue to grow with smaller values of $\alpha$.

\subsection{Simulations With Negative Correlation between $\pi$ and $u$} 

In Section \ref{generation}, $\pi$ and $u$ are mostly positively correlated, which is a reasonable assumption considering that families also care about finding a job. However, our simulations suggest that also in the case where $\pi$ and $u$ are mostly negatively correlated, CRSD and CRV can significantly improve family welfare.

\begin{table}[H]
\begin{center}
\begin{tabular}{ |l|c|c|c|c| }
\hline
\makecell{$\omega^\pi_{f\ell}$} & \makecell{$\ell_1$ (1)} & \makecell{$\ell_2$ (9)} & \makecell{$\ell_3$ (6)} & \makecell{$\ell_4$ (10)}\\
\hline
$f_1 (15)$ & 0.3 & 0.5 & 0.5 & 0.6\\
$f_2 (25)$ & 0.2 & 0.1 & 0.3 & 0.4\\
$f_3 (20)$ & 0.2 & 0.3 & 0.1 & 0.4\\
$f_4 (40)$ & 0.2 & 0.2 & 0.2 & 0.2\\
\hline
\end{tabular}
\end{center}
\caption {Values for generating predicted employment probabilities. For family $i$ of type $f$ and location $j$ of type $\ell$, $\pi_{ij} \sim \mathcal{U}(0, \omega^\pi_{f\ell})$.}\label{table:gen-pi-neg}
\end{table}

\begin{table}[H]
\begin{center}
\begin{tabular}{ |l|c|c|c|c| }
\hline
\makecell{$\omega^u_{f\ell}$} & \makecell{$\ell_1$ (1)} & \makecell{$\ell_2$ (9)} & \makecell{$\ell_3$ (6)} & \makecell{$\ell_4$ (10)}\\
\hline
$f_1 (15)$ & 1.0 & 0.6 & 0.6 & 0.3\\
$f_2 (25)$ & 0.8 & 1.0 & 0.6 & 0.3\\
$f_3 (20)$ & 0.8 & 0.6 & 1.0 & 0.3\\
$f_4 (40)$ & 1.0 & 0.6 & 0.6 & 0.3\\
\hline
\end{tabular}

\end{center}
\caption {Values for generating preference orders. For family $i$ of type $f$ and location $j$ of type $\ell$, $u_{ij} \sim \mathcal{U}(0, \omega^u_{f\ell})$.}\label{table:gen-val-neg}
\end{table}

\begin{figure}[H]
\centering
\includegraphics[scale=0.25]{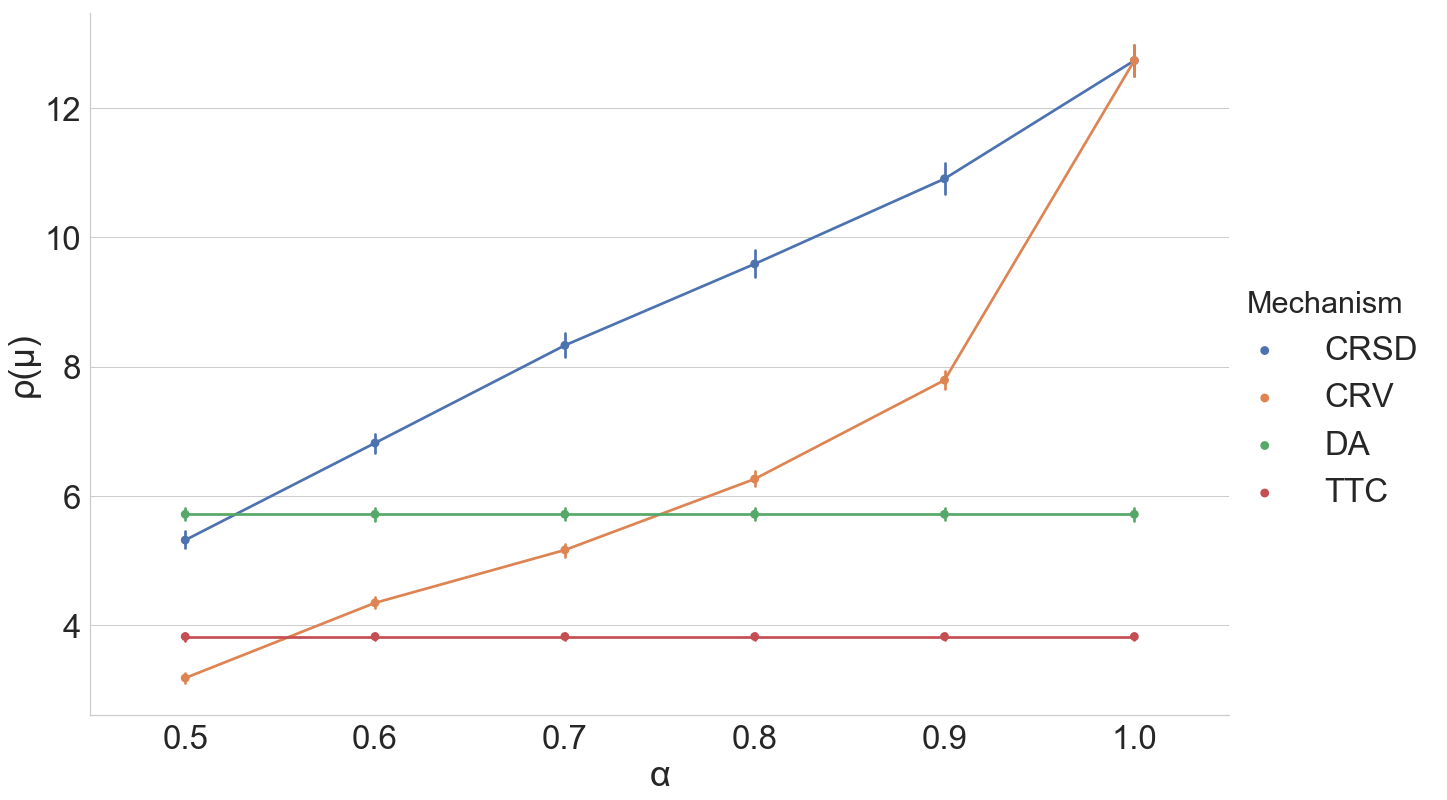}
\caption{Performance of CRSD, CRV, TTC and DA in terms of family welfare.}
\label{fig:average_rank_complete_neg_cor}
\end{figure}

\begin{figure}[H]
\centering
\includegraphics[scale=0.25]{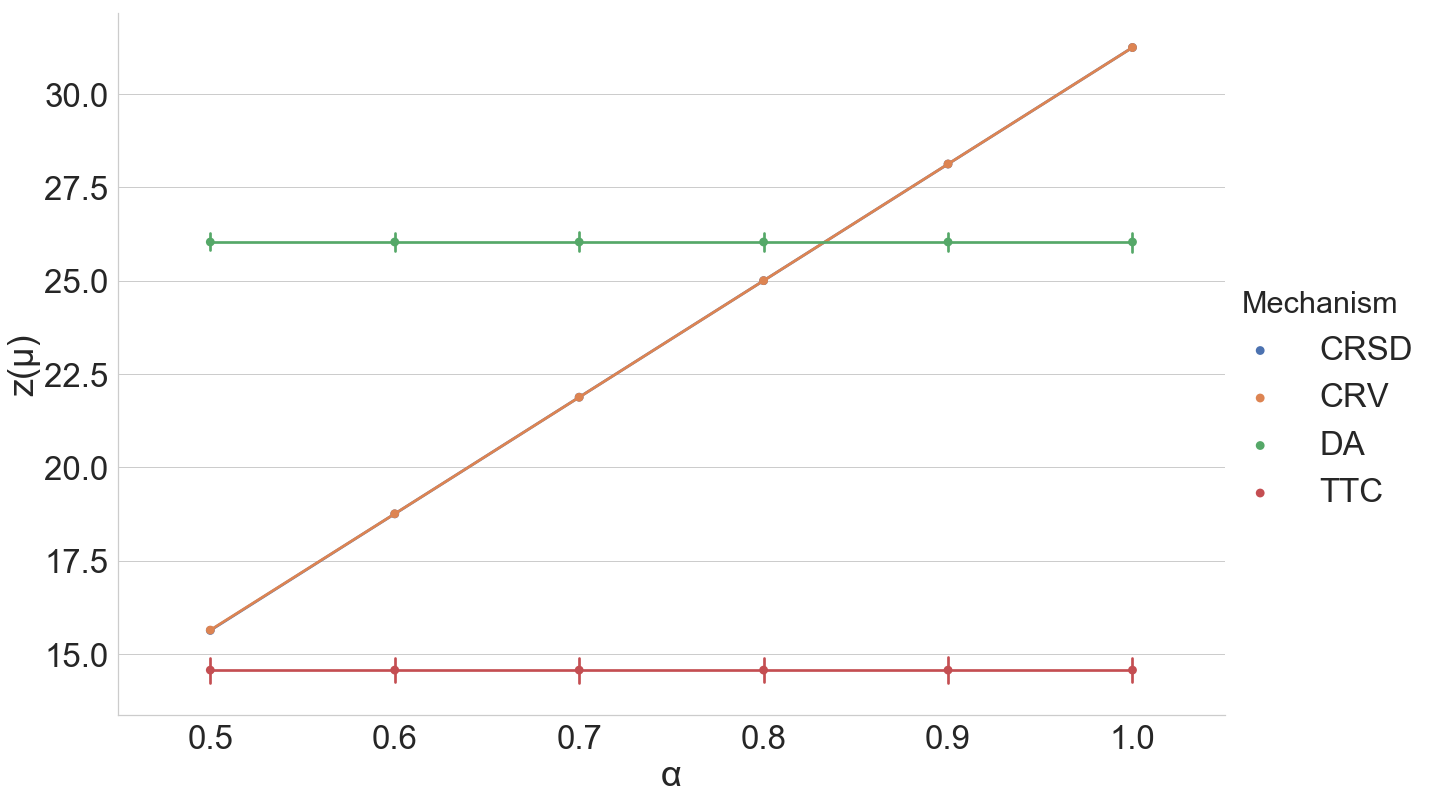}
\caption{Performance of CRSD, CRV, TTC and DA in terms of employment probability.}
\label{fig:employment_complete_neg_cor}
\end{figure}

This does not only hold for the setting with complete preference orders, but also for the setting with incomplete preference orders.

\begin{figure}[H]
\centering
\includegraphics[scale=0.25]{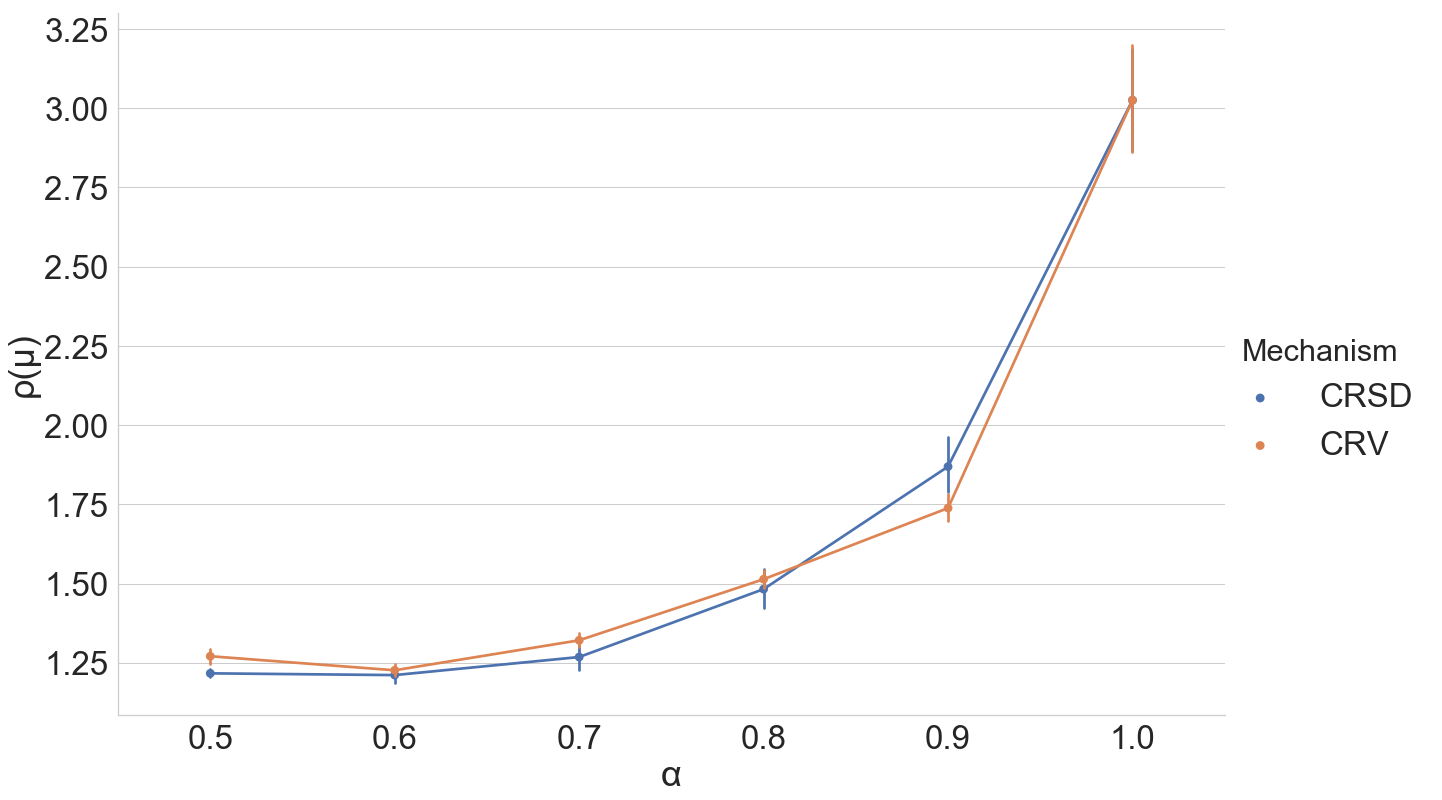}
\caption{$\rho(\mu)$ of CRSD and CRV for different values of $\alpha$.}
\label{fig:average_rank_incomplete_neg_cor}
\end{figure}

\begin{figure}[H]
\centering
\includegraphics[scale=0.25]{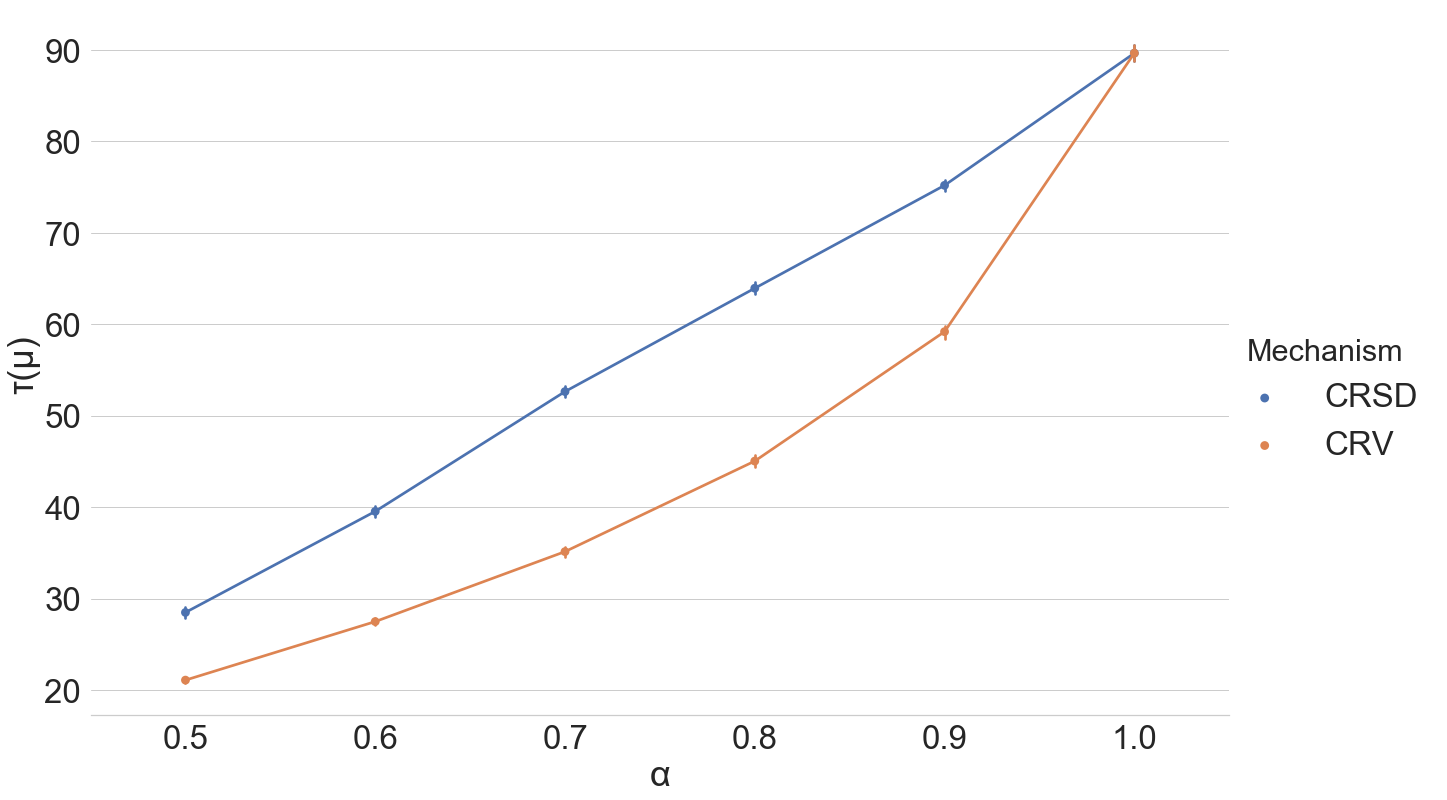}
\caption{$\tau(\mu)$ of CRSD and CRV for different values of $\alpha$.}
\label{fig:employment_incomplete_neg_cor}
\end{figure}

\end{document}